\documentclass[runningheads]{llncs}
\usepackage[T1]{fontenc}
\usepackage{amsmath,graphicx} 
\spnewtheorem{numclaim}[theorem]{Claim}{\bfseries}{\itshape}
\spnewtheorem{observation}[theorem]{Observation}{\bfseries}{\itshape}

\usepackage{fonttable}

\newcommand\squigglea{%
{\usefont{U}{lasy}{m}{n}\char58\char58}
}
\newcommand\squiggleLR{%
{\usefont{U}{lasy}{m}{n}\char58\char59}
}

\usepackage{verbatim}

\usepackage{thmtools,hyperref}
\usepackage{thm-restate}
\usepackage{lineno}
\newcommand\correspondingauthor{\thanks{Corresponding author.\\A preliminary version of the paper appeared at the 50th International Workshop on Graph-Theoretic Concepts in Computer Science (WG 2024)~\cite{wg24}.}}

\usepackage{tikz}
\usetikzlibrary{snakes}
\usepackage{graphicx,amssymb,amsmath}
\usepackage{xspace}
\usepackage{wasysym}
\usepackage{subcaption}
\usepackage[disable]{todonotes}
\newcommand{\changed}[1]{{\color{black}{#1}}}

\newcommand{\tb}[1]{\todo[color=blue!50!white]{#1}}

\newcommand{\discuss}[1]{\todo[inline,color=red!70!white]{Discuss (TB$\rightarrow$DM): #1}}
\newcommand{\calT}{\ensuremath{{\cal T}}}

\newcommand{\leftsquigarrow}{\raisebox{1ex}{\rotatebox{180}{$\,\:\text\squiggleLR\!\,$}}}
\newcommand{\girth}{\ensuremath{\text{fence-girth}}\xspace}
\newcommand{\rad}[1]{\text{\rm rad}(#1)}
\newcommand{\diam}[1]{\text{\rm diam}(#1)}
\renewcommand{\rho}{\ensuremath{r}}
\newcommand{\myalpha}{{\color{black}\ensuremath{\mathbf{a}}}}

\begin{document}
\title{Improved Outerplanarity Bounds for Planar Graphs}
%
%

\author{Therese Biedl\correspondingauthor\inst{1}
\and
Debajyoti Mondal\inst{2}
}
\authorrunning{Biedl and Mondal}
 
\institute{David R. Cheriton School of Computer Science, University of Waterloo, Canada   
\email{biedl@uwaterloo.ca}   \and
Department of Computer Science, University of Saskatchewan, Canada 
\email{d.mondal@usask.ca}}

\maketitle              
\begin{abstract}
In this paper, we study the outerplanarity of planar graphs, i.e., the number of times that we must (in a planar embedding that we can initially
freely choose) remove the outerface vertices until the graph is empty. It is well-known that there are $n$-vertex graphs with outerplanarity $\tfrac{n}{6}+\Theta(1)$, and not difficult to show that the outerplanarity can never be bigger. We give here improved bounds of the form $\tfrac{n}{2g}+2g+O(1)$, where $g$ is the \emph{fence-girth}, i.e., the length of the shortest cycle with
vertices on both sides.
This parameter $g$ is at least the connectivity of the graph, and often bigger; for example, our results imply that planar bipartite graphs have outerplanarity $\tfrac{n}{8}+O(1)$. We also show that the outerplanarity of a planar graph $G$ is at most $\tfrac{1}{2}\diam{G}+O(\sqrt{n})$, where $\diam{G}$ is the diameter of the graph.   All our bounds are tight up to smaller-order terms, and a planar embedding that achieves the outerplanarity bound can be found in linear time.
\end{abstract} 
 
\discuss{Some renaming of variables should still happen to be consistent.    \newline
* $g$ should perhaps be greek ($\gamma$? or something that captures fence, maybe $\phi$ J: $\varphi$ looks nicer, $\phi$ reminds me of empty set - should I replace?).  T: At this point I have gotten so used to $g$ that $\varphi$ (or perhaps $\psi$) might confuse me.  Let me mull this over longer. \newline
* We're fairly inconsistent whether to write $g/2$, $\tfrac{1}{2}g$ or $\tfrac{g}{2}$. 
My suggestion would be to switch everything to $\tfrac{g}{2}$ when the numerator/denominator have no sub/superscripts and to $\tfrac{1}{2}g^*$ otherwise. 
I'm changing this as I'm finding them. 
}

\discuss{Some English questions. \newline
* outerplanarity or outer-planarity?
\newline
* outerface or outer-face (or even outer face)?
All of the above are correct in some style-guides, so I don't really care but we should be consistent.   Without dash or space is shortest.
}
\section{Introduction}

The \emph{outerplanarity} of a planar graph is a well-known tool, both for deriving efficient algorithms and for proving lower bounds for graph drawings.
It measures how often we have to remove the vertices on the outerface (a \emph{peel}) until the graph is empty.   (Detailed definitions are in Section~\ref{sec:definitions}.)
In this paper, we obtain better upper bounds on the outerplanarity of a planar graph, which is important from the perspective of the following two application areas.

The first application of outerplanarity is to design faster algorithms for various problems in planar graphs. Baker~\cite{DBLP:conf/focs/Baker83} showed that for a planar graph with constant outerplanarity, numerous graph problems, such as independent set, vertex cover, dominating set 
can all be solved in linear time. 
(There are numerous generalizations, see e.g.~%
\cite{DBLP:journals/jgaa/Eppstein99,DBLP:journals/algorithmica/Eppstein00,DBLP:journals/jcss/DemaineHNRT04,DBLP:journals/combinatorica/Grohe03,DBLP:journals/endm/HajiaghayiNRT01}.\changed{)}
The running times of many such algorithms have an exponential dependency on the outerplanarity or related parameters. Hence an upper bound on the outerplanarity  with respect to the size of the graph can provide an estimate of how large of a graph these algorithms may be able to process in practice.

Another major application of outerplanarity 
is to derive lower bounds for various optimization criteria in graph drawing. For example, there exists a planar graph with a fixed planar embedding (known as nested triangles graph) that requires at least a $\tfrac{2}{3}n\times \tfrac{2}{3}n$-grid in any of its straight-line grid drawings that respect the given embedding (attributed to Leiserson \cite{DBLP:conf/focs/Leiserson80} by Dolev, Trickey and Leighton \cite{dolev1983planar}). Here a grid drawing maps each vertex to a grid point and each edge to a straight line segment between its end vertices. The crucial ingredient to their proof is that the nested triangles graph has $\tfrac{n}{3}$ peels (in this embedding), and any embedding-preserving planar straight-line grid-drawing of a planar graph with $k$ peels requires at least a $2k\times 2k$-grid. 
(In fact, this lower bound holds for many other planar graph drawing styles~\cite{DBLP:conf/gd/AlamBRUW15,DBLP:journals/jgaa/GiacomoDLM05,DBLP:journals/ipl/ZhangH05}.)  
Nested triangles graphs have outerplanarity $\tfrac{n}{6}$, and thus gives a lower bound of an $\tfrac{n}{3}\times \tfrac{n}{3}$-grid for the planar straight-line grid drawing even when one can freely choose an embedding to draw the graph. 
This raises a natural question of whether $\tfrac{n}{6}$ is the largest outerplanarity (perhaps up to lower-order terms) that a planar graph can have. This turns out to be true, via a detour into the radius, which we discuss next.

The \emph{eccentricity} of a vertex $v$ in $G$ is the smallest integer $k$ such that the  shortest-path distance from $v$ to any other vertex in $G$ is at most $k$. The \emph{radius} of $G$ (denoted $\text{rad}(G)$) is the smallest eccentricity over all the vertices of $G$, 
while the \emph{diameter} of $G$ (denoted $\text{diam}(G)$) is the largest eccentricity.  For 3-connected planar graphs, Harant~\cite{harant1990upper} proved an upper bound of $\rad{G}\leq \frac{n}{6}+\Delta^*+ \frac{3}{2}$, where $\Delta^*$ is the maximum degree of the dual graph, i.e., the maximum length of a face.
Ali et al.~\cite{DBLP:journals/dm/AliDM12} improved the upper bound to  $\frac{n}{6}+\frac{5\Delta^*}{6}+ \frac{5}{6}$ and 
more generally $\frac{n}{2\kappa}+O(\Delta^*)$, where $\kappa$ is the connectivity of the graph; these bounds are tight within an additive constant. 
This easily implies upper bounds on the outerplanarity.

\begin{restatable}{observation}{SimpleBound}
\label{obs:simpleBound}
Every planar graph $G$ has outerplanarity at most \changed{$\min\{1 + \rad{G}, \tfrac{n+26}{6}\}$}, and this bound holds even if the spherical embedding of $G$ is fixed.
\end{restatable}
\begin{proof}
We first prove the radius-bound.  Use as outerface a face that is incident to a vertex $v$ of eccentricity $\rad{G}$.
Then all vertices $z$ with $d_G(v,z)=i-1$ belong to the $i$th peel or an earlier one, so after removing $\rad{G}+1$ peels the graph is empty.

For the second bound, arbitrarily add edges to $G$ to make it into a maximal planar graph $G^+$.   This is triangulated, so using the result by Ali et al.~\cite{DBLP:journals/dm/AliDM12} we have $\rad{G^+}\leq \tfrac{n+20}{6}$ and hence outerplanarity at most $\tfrac{n+26}{6}$.    
The outerplanarity of subgraph $G$ cannot be bigger.
\qed\end{proof}

A \emph{triangulated graph} $G$ is a maximal planar graph; in any planar embedding, faces then have length 3.  It is folklore that for a triangulated graph, 
the difference between radius and outerplanarity is at most~1. 
But for graphs with greater face-lengths, the two parameters become very different (consider a cycle). The radius-bound on 3-connected graphs by Ali et al.~\emph{increases} as the faces get bigger, while one would expect the outerplanarity to \emph{decrease} as faces get bigger.    So our goal in this paper is to find bounds on the outerplanarity that do not depend on the face-lengths and improve on $\tfrac{n}{6}$ for some graphs.  We use a parameter that we call the 
\emph{fence-girth}:   
In a planar graph $G$ with a fixed embedding, a \emph{fence} is a cycle $C$ \changed{with} other vertices \changed{both strictly inside and strictly outside $C$},
and the \emph{fence-girth} is the shortest length of a fence.   \changed{(For a graph without cycles, the fence-girth is $\infty$.)}
Our main result is the following:

\begin{enumerate} 
    \item[] {\bf C1.} 
Every planar graph has outerplanarity at most $\lfloor \tfrac{n-2}{2g} \rfloor +O(g)$ for any integer $g\geq 3$ that is at most the fence-girth.
We can find a planar embedding with this number of peels in linear time.  
Some graphs with \girth $g$ have outerplanarity at least $\lfloor \tfrac{n-2}{2g} \rfloor$.   
(Section~\ref{sec:fencegirth}).
\end{enumerate}


We are not aware of prior results for outerplanarity-bounds, but since the radius is closely related to it for triangulated graphs, we 
contrast our result to the best radius-bound of $\tfrac{n}{2\kappa}+O(\Delta^*)$ by Ali et al.~\cite{DBLP:journals/dm/AliDM12}.     
The fence-girth is never less than the connectivity $\kappa$,
so our theorem implies outerplanarity $\tfrac{n}{2\kappa}+O(1)$\changed{, where $\kappa\leq 5$}. Hence up to small constant terms our bound is never worse than~Ali et al.'s, and often it will be better.   
For example, for bipartite planar graphs the fence-girth is at least 4, so with $g=4$ we obtain a bound of $\tfrac{n}{8}+O(1)$, whereas  Ali et al.'s bound is only $\tfrac{n}{6}+O(1)$. 
Secondly, the prior bound held only for 3-connected planar graphs, while we make no such restrictions.
Finally, we can find a suitable embedding in linear time while   
all existing algorithms for outerplanarity~\cite{DBLP:journals/algorithmica/BienstockM90,DBLP:conf/esa/Kammer07} 
take quadratic time or more.

\medskip

For a triangulated graph $G$, the fence-girth is the same as the connectivity $\kappa\leq 5$.
Result {\bf C1.} hence implies that $\rad{G}\leq \tfrac{n}{2\kappa}+O(1)$.   This bound was previously known \cite{DBLP:journals/dm/AliDM12},
but our result comes with a linear-time algorithm to find a vertex with this eccentricity, which is new:

\begin{enumerate} 
    \item[] {\bf C2.} For a $\kappa$-connected triangulated graph $G$, we can find a vertex $s$ with  $d_G(s,z)\leq \lfloor \tfrac{n-2}{2\kappa} \rfloor +O(1)$ for all $z\in V(G)$ in linear time (Section~\ref{sec:outerplanarity}). 
\end{enumerate}

Ali et al.~did not study the run-time to find a vertex of small eccentricity; while their proof could be turned into an algorithm, its run-time would be $O(n\cdot \rad{G})$, hence quadratic. 
The known subquadratic algorithms for computing the radius of a planar graph are far from being linear~\cite{cabello2018subquadratic,DBLP:journals/siamcomp/GawrychowskiKMS21,wulff2008wiener}, and algorithms that provide $(1+\epsilon)$-approximation have running time of the form $O(f(1/\epsilon) n\log^2 n)$~\cite{DBLP:journals/algorithmica/ChanS19,DBLP:journals/talg/WeimannY16}, where $f$ is a polynomial function on $(1/\epsilon)$. Linear-time algorithms for the radius are only known for special subclasses of planar graph classes~\cite{DBLP:conf/soda/ChepoiDV02,DBLP:journals/jgaa/Eppstein99}. 
\medskip 

Since the outerplanarity (for triangulated graphs) is closely related to the radius, and the radius is closely related to the diameter, it is natural to ask to bound the outerplanarity in terms of the diameter.   We can show the following:

\begin{enumerate}
    \item[] {\bf C3.} Every planar graph $G$ has outerplanarity at most $\tfrac{1}{2}\diam{G}+O(\sqrt{n})$, and a corresponding embedding can be found in linear time. 
Every triangulated graph $G$ has radius $\tfrac{1}{2}\diam{G}+O(\sqrt{n})$, and a vertex of this eccentricity can be found in linear time.   
(Section~\ref{sec:diameter}).
\end{enumerate}

Similar results with a `correction term' of $O(\sqrt{n})$ have been studied before, for example, Boitmanis et al.~\cite{DBLP:conf/wea/BoitmanisFLO06} gave an algorithm that computes the diameter and radius within such an error term in $O(|E(G)|\sqrt{n})$ time.
So our contribution is that
we can find a vertex of eccentricity $\tfrac{1}{2}\diam{G}+O(\sqrt{n})$ in linear time, hence faster than Boitmanis et al.~\cite{DBLP:conf/wea/BoitmanisFLO06}. 

We also show that this bound is tight and that the correction-term $O(\sqrt{n})$ cannot be avoided, not even for triangulated graphs.
In particular, this answers (negatively)
a question on \texttt{MathOverflow} \cite{mathoverflow} whether $\rad{G}\leq \tfrac{1}{2}\diam{G} + O(1)$ for all triangulated graphs;
such a relationship does hold for
interval graphs \cite{Pramanik2011}, chordal graphs \cite{shook2022characterization}, and various grid graphs and generalizations \cite{DBLP:conf/soda/ChepoiDV02}.

\begin{enumerate}
	\item[] {\bf C4.} 
There exists a triangulated graph $G$ with radius $\tfrac{1}{2}\diam{G}+\Omega(\sqrt{n})$ (Section~\ref{sec:diameter}).
\end{enumerate}	

\section{Definitions}
\label{sec:definitions}

We assume familiarity with graph theory and planar graphs (see 
for example \cite{Die12}) and fix throughout a planar graph $G$ with $n$ vertices.
For a path $\pi$ in $G$, the \emph{length} $|\pi|$ is its number of edges.  
For two vertices $y,z$, write $d_G(y,z)$ for the length of the shortest path between them;  \changed{we only need undirected graph distance, i.e.,} if $G$ has directed 
edges then this measures the distance in the underlying undirected graph.
For a set of vertices $L$, write $G\setminus L$ for the graph obtained by deleting the vertices in $L$
and $G[L]:=G\setminus (V\setminus L)$ for the graph induced by $L$.
We need the following \emph{separator theorem for trees}:

\begin{theorem}\cite{LT79} 
\label{thm:LT}
Let $\calT$ be a tree with non-negative node-weights $w(\cdot)$.
Then in linear time we can find a node $S$ such   that
for every subtree $\calT'$ of $\calT\setminus S$ we have $w(\calT')\leq \tfrac{1}{2}w(\calT)$,
where $w(\calT')$
denotes the sum of weights of nodes in $\calT'$.
\end{theorem}

One easily derived consequence of the separator theorem is the following:

\begin{restatable}{observation}{RadiusUpper}
\label{obs:radiusUpper}
Any connected graph $G$ has a vertex with eccentricity
at most $\lfloor \tfrac{n}{2} \rfloor$ that we can find in linear time.
\end{restatable}
\begin{proof}
    Fix a spanning tree $\calT$ of the graph, and let $s$ be the separator-node from Theorem~\ref{thm:LT}, using unit weights.   Then any subtree $\calT'$ of $\calT\setminus \{s\}$ contains at most $\lfloor n/2 \rfloor$ nodes, and so the distance from $s$ to any other node is at most $\lfloor n/2 \rfloor$.
\qed\end{proof}

A {\em spherical embedding} of $G$
describes a drawing $\Gamma$ of $G$ on a sphere $\Sigma$
by listing for each \emph{face} (maximal region of $\Sigma\setminus \Gamma$) the \changed{closed walk}(s)
of $G$ that bound the face.   Graph $G$ is called \emph{triangulated} if all faces are
triangles; the spherical embedding is then unique.  A \emph{planar embedding} of $G$ is 
a drawing of $G$ in the plane described by giving a spherical embedding $\Gamma$ and fixing one face $F$
(the \emph{outerface}) which becomes the infinite face in the planar drawing.

For the following definition, assume that $G$ is \emph{plane} (comes with a fixed planar embedding $(\Gamma,F)$). Define
the {\em peels}~\cite{DBLP:conf/esa/Kammer07} of $G$ as follows:
$L_1$ consists of all vertices on the outerface $F$.   For $i>1$,
$L_i$ consists of all vertices on the outerface of 
$G\setminus (L_1{\cup}\dots{\cup}L_{i-1})$, where this graph uses as
planar embedding the one inherited from $G$.   
The \emph{number of peels} (which depends on $\Gamma$ and $F$) is the minimum number $k$ such that
$G\setminus (L_1{\cup}\dots{\cup}L_{k})$ is the empty graph.
We use the term \emph{fixed-spherical-embedding (fse) outerplanarity} of $G$ for the minimum number of peels
over all choices of outerface $F$ (but keeping the same spherical embedding $\Gamma$).
The (unrestricted) \emph{outerplanarity} of $G$ is the minimum number of peels over
all choices of spherical embedding $\Gamma$ and outerface $F$ of $\Gamma$.


\section{Toolbox}
\label{sec:preliminaries}
\label{sec:toolbox}

In this section, we give some definitions and methods that will be used by multiple proofs later.
Throughout, we assume that the input graph $G$ 
comes with a fixed spherical embedding which we will never change.   
We also assume that $G$ is connected, for if it is not then we can add
edges between components that share a face until $G$ is connected.
This does not add cycles (so does not change the fence-girth) and it can only
decrease the diameter (hence improve the outerplanarity bound),  therefore adding such edges does not affect our results.

\paragraph{The tree of peels:}
We will compute a tree $\calT$ that stores, roughly speaking, the hierarchy of peels for some outerface,
see also Figure~\ref{fig:mesh2}.
Formally, pick a \emph{root-vertex} $\rho$ arbitrarily, except that it should not be a cutvertex. Choose as outerface of $G$ a face incident to $\rho$. 
Define $L_0:=\{\rho\}$ and compute the peels $L_1,\dots,L_k$ of $G\setminus L_0$.
These \emph{layers} $L_0,L_1,\dots,L_k$ are not quite the peels of $G$ (because we start with one vertex rather than a face), and not quite the layers of a breadth-first search (BFS) tree (because we include in the next layer all vertices that share a face with vertices of the previous layer, whether they are adjacent or not).
We direct each edge $(y,z)$ 
from the higher-indexed to the lower-indexed layer; edges connecting vertices within a layer remain undirected.

We organize the layers $L_0,\dots,L_k$ into a tree $\calT$ (the \emph{tree of peels}) as follows:
The root of $\calT$ is a node $R$ that corresponds to the entire graph $G$; define $V(R):=\{\rho\}$.
For $i=1,2,3,\dots$, add a node $N_K$ to $\calT$ for each connected
	component $K$ of $G\setminus (L_0{\cup}\dots{\cup} L_{i-1})$.
	Component $K$ is part of one connected component $P$ of
	$G\setminus (L_0{\cup}\dots{\cup} L_{i-2})$; make $N_K$ a child
	of the node $N_P$ corresponding to $P$ in $\calT$.   Define $V(N_K)$ to be the outerface vertices of $K$.

\begin{figure}[ht]
\subcaptionbox{~}{{\includegraphics[scale=0.65,trim=0 0 407 0,clip]{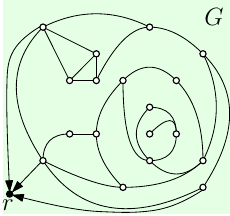}}}%
\hspace*{\fill}%
\subcaptionbox{~}{{\includegraphics[scale=0.65,trim=130 0 278 0,clip]{pictures/layertree.pdf}}}%
\hspace*{\fill}%
\subcaptionbox{~}{{\includegraphics[scale=0.7,trim=265 0 180 0,clip]{pictures/layertree.pdf}}}%
\hspace*{\fill}%
\subcaptionbox{~}{{\includegraphics[scale=0.7,trim=362 0 128 0,clip]{pictures/layertree.pdf}}}%
\hspace*{\fill}%
\subcaptionbox{~}{{\includegraphics[scale=0.7,trim=402 0 90 0,clip]{pictures/layertree.pdf}}}%
\hspace*{-3mm}
\subcaptionbox{~}{{\includegraphics[scale=0.65,trim=445 0 2 0,clip]{pictures/layertree.pdf}}}%
\hspace*{\fill}%
\subcaptionbox{~\label{fig:gplus}}{\includegraphics[scale=0.5,trim=0 175 0 0,clip]{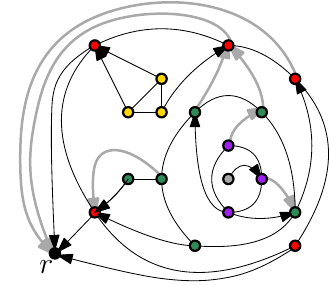}}
    \caption{(a) A plane graph $G$. (b)--(e) The graphs obtained by deleting  $L_0,L_1,\ldots, L_3$, respectively. Solid vertices are the set $V(N)$ of the corresponding node $N$. 
    (f) The tree of peels $\cal T$. (g) The augmentation $H$.
	}
    \label{fig:mesh2}
\end{figure}

Throughout this paper, we will use `node' (and upper-case letters) for the elements of $\calT$ while we reserve `vertex' (and lower-case letters) for $G$.
An \emph{interior node} of $\calT$ is a node that is neither the root nor a leaf.
We think of each node $N$ as `storing' the vertices in $V(N)$ and observe that these vertices induce a connected subgraph.
Also, every vertex of $G$ is stored at exactly one node of $\calT$.
We need a few easy observations: 

\begin{restatable}{observation}{NodeSize}
\label{obs:nodeSize}
\label{obs:Tproperty}
The following holds for the tree of peels $\calT$:
\begin{enumerate}
\item The root $R$ of $\calT$ has a single child.
\item
\label{it:Tproperty}
Let $Y,Z$ be the nodes that store the ends $y,z$ of an edge $e$. 
Then either $e$ is undirected and $Y{=}Z$, or $e$ is directed (say $y\rightarrow z$) and $Z$ is the parent of $Y$.
\item 
\label{it:nodeSize} For any interior node $N$ of $\calT$, the size $|V(N)|$ is at least the \girth.
\end{enumerate}
\end{restatable}
\begin{proof}
(1) We chose root-vertex $\rho$ so that it is not a cutvertex; therefore $G\setminus L_0=G\setminus \rho$
is connected and there is only one node that is a child of $R$.

(2) For edge $(y,z)$, let $i$ be the smallest index for which layer $L_i$ contains $y$ or $z$. If both $y,z$ are in $L_i$, then $y,z$ belong to the same node since they are in one connected component.   If one of them (say $y$) is not in $L_i$, then $y\in L_{i+1}$, so the edge is directed $y\rightarrow z$ and $Y$ becomes a child of $Z$ since the edge $(y,z)$ ensures that $y$ is in the connected component that defined $Z$.


(3) Recall that node $N$ corresponds to a connected component $K_N$ of the graph obtained by deleting some of the levels. Since $N$ is not a leaf, subgraph $K_N$ has at least one vertex $v$ not on the outerface.  
Therefore, the outerface of $K_N$ contains a cycle that has $v$ inside and $\rho$ outside (since $N\neq R$).
This cycle is a fence and all its vertices belong to $V(N)$.
\qed\end{proof}

\paragraph{Augmenting $G$:}  It will be helpful if every vertex except root-vertex $\rho$ has an outgoing edge.
In general, this need not hold for our input graph $G$.   We therefore augment $G$ with further edges.
The following result was shown in \cite{Bie15}; the result there was for
the peels while our definition of layers $L_0,\dots,L_k$ is slightly different,
but one easily verifies that the proof carries over.

\begin{numclaim}(based on Obs.~2 in \cite{Bie15})
\label{claim:parents}
We can add edges to $G$ (while maintaining planarity) such that for all $i \geq 1$ every vertex in $L_i$ has a neighbour in $L_{i-1}$.
\end{numclaim}

Let $H$ be the graph obtained by adding a set of directed edges 
such every vertex except $\rho$ has an outgoing edge in $H$ (Figure~\ref{fig:gplus}).
Because we only add edges between adjacent layers, the following is easily shown.

\begin{restatable}{observation}{AugmentedT}
\label{obs:augmentedT}
The augmented graph $H$ has the same tree of peels as $G$ (assuming we start with the inherited planar embedding and outerface and use the same root-vertex). 
\end{restatable}
\begin{proof}
Observe first that since we start with the same root-vertex and outerface, and only add directed edges, 
\tb{FYI: We used to assume that $H$ is minimal, but I dont want to waste effort on discussing how to find this efficiently.   So downgraded this to `add only directed edges', which is easy to achieve.}
the layers $L_0,\dots,L_k$ are exactly the same in both $G$ and $H$.  Assume $H$ is obtained by adding just one edge $e=(y,z)$ (the full proof is then by induction on the number of added edges).   
Observe that $y,z$ belong to one face $F$ of $G$ since we can add edge $(y,z)$ to $G$ while staying planar.   
Since $G$ is connected, so is the boundary of $F$.   For any face the vertices belong to at most two consecutive layers by definition of peels; 
for the specific face $F$ the vertices belong to exactly two consecutive layers (say $L_i$ and $L_{i+1}$) since they include $y,z$ which are not on the same 
layer by construction of $H$.   So there exists in $G$ a path $\pi$ (along the boundary of $F$) that connects $y$ and $z$ and has all vertices in $L_i$ and $L_{i+1}$.

With this, the connected components that define the tree of peels are exactly the same for both graphs.   To see this, observe that when we have removed $L_0\cup \dots \cup L_{h-1}$ for some $h\leq i$, then edge $(y,z)$ provides no connectivity among components that we did not have via path $\pi$ instead.   Once we have removed $L_i$, one of $y,z$ (and hence the added edge) has been removed from the graph, and so again does not add any connectivity.
\qed\end{proof}

\begin{figure}[ht]
    \centering
    \includegraphics[width=.8\textwidth]{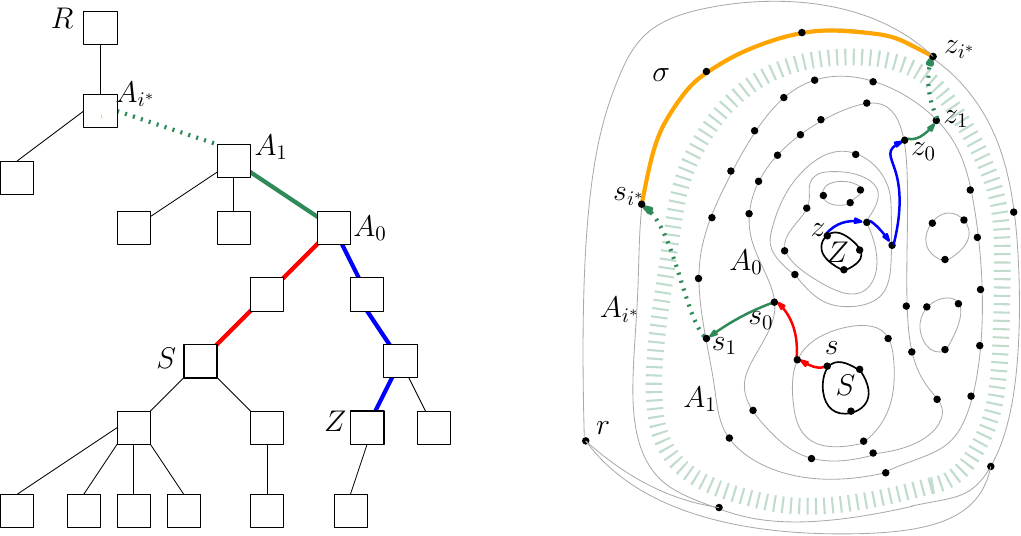}
    \caption{Larger example of $\cal T$ and the detour-method to connect $s_0\in A_0$ with $z_0\in A_0$, as well as $s$ to $z$ as in 
Claim~\ref{cl:caseRoot}. Not all directed edges are shown for clarity.}
    \label{fig:pathfinding}
    \label{fig:detourMethod}
\end{figure}

Adding edges to $G$ 
may decrease the \girth, but not the node-sizes of
$\calT$,  which is all that will be used below.
We will in the following \emph{only} consider graph $H$, so 
every vertex $z\neq \rho$ has an outgoing edge.   

\paragraph{The detour-method:} 
We need the following method to connect two given vertices of $H$ that are stored in the same node of $\calT$
(Figure~\ref{fig:method1} and Figure~\ref{fig:detourMethod}).

\begin{definition}
Fix a node $A_0$ of 
$\calT$, and two vertices $s_0,z_0\in V(A_0)$, as well
as \emph{exit conditions} $\xi_0,\xi_1,\xi_2,\dots$ which are (possibly negative) integers
\changed{that will be specified by each application}.  The \emph{detour method at node $A_0$}
finds a path connecting $s_0$ and $z_0$ as follows:
\begin{itemize}
\item Initialize $i=0$; we have $s_i,z_i\in V(A_i)$.  We will
	also maintain paths $\tau_s$ (from $s_0$ to $s_i$) and $\tau_z$ (from $z_0$ to $z_i$);
	initially these are simply $\langle s_0\rangle$ and $\langle z_0\rangle$.
\item If $d_H(s_i,z_i)\leq \xi_i$, then set $\sigma$ to be a path from $s_i$ to $z_i$ that has distance at most $\xi_i$. Exit
	with `success' and return $i$ and $\tau_s,\tau_z,\sigma$.
\item If $d_H(s_i,z_i)> \xi_i$, and $A_i$ was the root, then return `fail'.
\item  Otherwise let $A_{i+1}$ be the parent of $A_i$.
	Find directed edges $s_i\rightarrow s_{i+1}$ and
	$z_i\rightarrow z_{i+1}$, append them to paths $\tau_s$ and $\tau_z$,  update $i\gets i{+}1$ and repeat.
\end{itemize}
\end{definition}

Since each iteration gets us closer to the root, the algorithm must terminate.
If it exits successfully, say at index $i^*$, then we get a walk
$s_0 {\overset{\tau_s}{\:\text\squiggleLR\!}} s_{i^*} {\overset{\sigma}{\:\text\squigglea\!
}} z_{i^*} {\overset{\tau_z}{\leftsquigarrow}} z_0.$
To bound the length of this walk, the following observation 
will be useful:
\begin{restatable}{observation}{PathWithinNode}
\label{obs:pathWithinNode}
If the detour-method does not succeed at index $i>0$, and $A_i$ is not the root, then
$\lfloor \tfrac{1}{2}|V(A_i)| \rfloor \geq \xi_i{+}1$.
\end{restatable}
\tb{proof kicked to appendix}
\begin{proof}
Recall that there are directed edges $s_{j-1}\rightarrow s_j$ and $z_{j-1}\rightarrow z_j$ by $j>0$ and that $s_{j-1},z_{j-1}\in A_{j-1}$.
Write $G_j$ for the graph induced by $V(A_j)$, and observe that $V(A_{j-1})$ must all belong to one inner face $F$ of $G_j$ since it bounds
a connected component of the subgraph of $G$ where the peels up to $A_j$ have been removed.   Furthermore, $s_j,z_j$ have neighbours both strictly inside
and strictly outside $F$ (due to their outgoing edges).    So there must be a simple cycle $C$ along the boundary of $F$ that contains both $s_j$ and $z_j$.
Walking along the shorter side of $C$ hence gives a walk from $s_j$ to $z_j$ of length at most $\lfloor |C|/2 \rfloor \leq \lfloor |V(A_j)|/2 \rfloor$.
By $j<i^*$ the stopping-condition did not hold at $A_j$, so this implies $\lfloor |V(A_j)/ \rfloor > \xi_i$ and the result holds by integrality. 
\qed\end{proof}

We demonstrate how to use the detour-method with the following result that will be needed later.
For any node $N$, write $\myalpha(N)$ for the number of vertices stored at strict ancestors of $N$.

\begin{restatable}{lemma}{ConnectS}
\label{lem:distanceS}
\label{lem:connectS}
Assume that $|V(N)|\geq 3$ for all interior nodes $N$. 
For any node $A_0$ and any $s_0,t_0\in V(A_0)$ we have
$d_H(s_0,t_0)\leq  \max\{\allowbreak 2\lceil \sqrt{\myalpha(A_0)}\rceil {-}2, \allowbreak 4\}$.
\end{restatable}

\begin{proof} 
We are done if $A_0$ is the root $R$ or its child or a grandchild of $R$, for then we can connect $s_0$
and $t_0$ with a path of length at most 4 by following directed edges until we reach root-vertex $r$.
So assume that $A_0$ has a parent $P$, grand-parent $G_1$ and great-grandparent $G_2$, 
therefore $\myalpha(A_0)\geq |V(P)|+|V(G_1)|+|V(G_2)| \geq 3+3+1=7$ since $P$ and $G_1$ are internal nodes.
Hence $\lceil \sqrt{\myalpha(A_0)} \rceil \geq 3$ and
the desired upper bound is $2\lceil \sqrt{\myalpha(A_0)} \rceil -2 \geq 4$.
For ease of writing define $\beta=\lceil \sqrt{\myalpha(A_0)} \rceil -3\geq 0$,
so the desired upper bound becomes $2\beta+4$.
Apply the detour-method at $A_0$, using 
$\xi_i= 2\beta+4-2i$ for $i\geq 0$.
If the method returns successfully at index $i^*$ with paths $\tau_s,\tau_z,\sigma$, then 
combining the three paths gives $d(s_0,t_0)\leq |\tau_s|+|\tau_z|+|\sigma| \leq 2i^* + 2\beta+4-2i^*$ as desired.

So now assume for contradiction 
that the detour-method fails, and let $A_i,s_i,z_i$ (for $i=0,\dots,\ell$) be the nodes and vertices that it used;  we have $A_\ell=R$ since the method 
can only fail at the root.
Since the detour-method did not succeed at $R$,
\tb{FYI: This needed a major rewrite; Obs.\ref{obs:pathWithinNode} doesn't hold for the root note.} 
we did not have a path of length at most $\xi_\ell$ from $s_\ell$ to $z_\ell$.   But $|R|=1$, so $s_\ell=r=z_\ell$ are connected by a path of length 0.    So $0>\xi_\ell=2\beta+4-2\ell$ 
or $\ell> \beta+2$.
Now bound the sizes of $V(A_1),\dots,V(A_\ell)$ as follows:
\begin{itemize}
\item For $i=1,\dots,\beta+2$, Observation~\ref{obs:pathWithinNode} implies 
	$|V(A_i)|\geq 2(\xi_i+1)=4\beta+10-4i > 3\beta+8-4i$.   
\item For $i=\beta{+}1$, this bound becomes $|V(A_i)|\geq 4\beta+10-4(\beta{+}1)=6$.
\item For $i=\beta{+}2$, we also know $|V(A_i)|\geq 3$ since $A_i$ is not the root by $\ell>\beta+2$.
\item Root $A_\ell$ stores one vertex $\rho$.
\end{itemize}
We therefore have a contradiction:
\begin{align*}
(\beta+3)^2 & \geq \myalpha(A_0) \geq \sum_{i=1}^\ell |V(A_i)|   
\geq \sum_{i=1}^{\beta} \big(3\beta{+}8{-}4i\big) + 6 + 3 + 1 \\
& > \sum_{i=1}^{\beta} \big(3\beta{+}6\big) + \sum_{i=1}^{\beta} \big(2{-}4i\big) + 9 \\
& = 3\beta^2+6\beta - 2\sum_{i=1}^{\beta} (2i{-}1) + 9
 = 3\beta^2+6\beta - 2\beta^2 +9  = (\beta+3)^2  
\end{align*}
\qed\end{proof}

\begin{figure}[t]
\subcaptionbox{~\label{fig:method1}}{\includegraphics[scale=0.37,page=1,trim=0 256 0 0,clip]{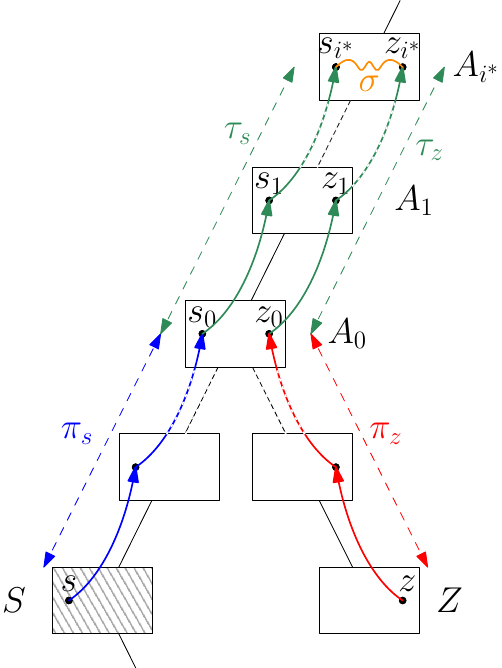}}
\hspace*{-5mm}
\subcaptionbox{~\label{fig:method2}}{\includegraphics[scale=0.37,page=2,trim=0 128 60 128,clip]{pictures/methods.pdf}}
\subcaptionbox{~\label{fig:method3}}{\includegraphics[scale=0.37,page=3,trim=0 0 0 256,clip]{pictures/methods.pdf}}
\subcaptionbox{~\label{fig:leastCommonAncestor}}{\includegraphics[scale=0.37,page=4,trim=0 0 0 256,clip]{pictures/methods.pdf}}
\caption{(a) Path-finding for nodes in $V(\calT_R)$ \changed{when $Z$ is not a descendent of $S$}. (b) Path-finding by using a detour at $S$ \changed{when $Z$ is a descendent of $S$}.
(c) Path-finding for deep descendants of $S$. (d) Finding a small common ancestor of deep nodes.
\todo[inline]{The green is barely visible in a B/W printout.}
	}
\label{fig:threeMethods}
\end{figure}

\section{Outerplanarity and \girth}
\label{sec:outerplanarity}
\label{sec:fencegirth}

In this section, we first prove that any planar graph $G$ has outerplanarity at most $\lfloor \tfrac{n-2}{2g}\rfloor +O(g)$
where $g\geq 3$ is a (user-given) integer that is \changed{supposed to be} at most the \girth.%
\footnote{%
We use a parameter $g$ that is separate from the \girth since the latter can be $\Theta(n)$; to minimize the upper bound one should set $g$ to be $\min\{\girth,\tfrac{1}{2}\sqrt{n{-}2}\}$.}
Then we discuss implications 
and lower bounds.

\paragraph{Separator-node $S$:}
To prove the upper bound, define the layers $L_0,L_1,\dots$, augmentation $H$ and  tree of peels $\calT$ as in Section~\ref{sec:preliminaries}.
\changed{If any interior node of $\calT$ stores fewer than $g$ vertices,
then repoert that $g$ was too big (cf.~Obs.~\ref{obs:Tproperty}(\ref{it:nodeSize})) and abort.  Otherwise, a}pply 
the separator theorem (Theorem~\ref{thm:LT}) to the tree $\calT$, 
using node-weight $w(N):=|V(N)|$, i.e., the number of stored vertices. 
\changed{Let $S$ be a node} such that any subtree $\calT'$ of $\calT\setminus S$ stores
at most $\tfrac{n}{2}$ vertices; we write $V(\calT')$ for these stored vertices.      We know that $S\neq R$,
because the root has only one child, and $w(\calT\setminus R)=n-1>\tfrac{n}{2}$.

Crucially, $S$ is close to all other nodes of $\calT$.   Since we will frequently
need the following upper bound, we introduce a convenient \changed{notation} $\delta$ for it.

\begin{observation}
\label{obs:distT}
For any node $Z$ of $\calT$ we have $d_\calT(S,Z)\leq \delta:=\lfloor \tfrac{n-2}{2g} \rfloor +1$.   
\end{observation}
\begin{proof}
Let $\Pi$ be the path from $S$ to $Z$ in $\calT$.
Every interior node $N$ of $\Pi$ 
belongs to the same subtree $\calT'$ of $\calT\setminus S$, and 
satisfies $|V(N)|\geq g$ since it is neither leaf nor root. 
Node $Z\neq S$ also belongs to $\calT'$ and $|V(Z)|\geq 1$.
Altogether therefore
$\tfrac{n}{2}\geq |V(\calT')|\geq g(|\Pi|{-}1)+1$ or $|\Pi|\leq \tfrac{n-2}{2g} + 1$,
which implies $d_{\calT}(S,Z)=|\Pi|\leq \delta$ by integrality.
\qed\end{proof}


The overall idea of our proof is now to pick a vertex $s\in V(S)$, and to argue that $d_H(s,z)\leq \delta + 2g-2$ for all vertices $z$.    Actually, for most cases below, this will hold for \emph{any} choice of $s\in S$.


\paragraph{Vertices stored in $\calT_R$:}
Let $\calT_R$ be
the subtree of $\calT\setminus S$ that contains root $R$.
For vertices in this subtree,
the detour method proves the distance-bound.

\begin{numclaim} 
\label{cl:caseRoot}
Assume that $g\geq 3$.
For any $s\in V(S)$ and any $z\in V(\calT_R)$ 
we have $d_{H}(s,z)\leq \delta + g$.
\end{numclaim}
\begin{proof}
Figure\changed{s~\ref{fig:detourMethod} and \ref{fig:method1}}
illustrate this proof.
Let $Z$ be the node that stores $z$, and let $A_0$ be the least
common ancestor of $S$ and $Z$; this is a strict ancestor of $S$ since $Z\in \calT_R$.   
Follow directed edges from $s$ to some vertex $s_0\in V(A_0)$; the resulting path $\pi_s$
has length $d_\calT(S,A_0)$ since every directed edge gets us closer to the root.
Likewise we can get a path $\pi_z$ of length $d_\calT(Z,A_0)$ from $z$ to some node $z_0\in V(A_0)$
(possibly $z_0=z$).

Now apply the detour-method with $A_0,s_0,z_0$, using $\xi_i=g{-}1$.   This will always exit with `success' at 
a node $A_{i^*}$ that is not the root,
because any two vertices stored at the child of the root have outgoing edges towards
the root-vertex $\rho$ and hence distance $2\leq g{-}1$.
Let $\tau_s,\tau_z,\sigma$ be the paths, then
$\pi:= s {\overset{\pi_s}{\:\text\squiggleLR\!}} s_0 {\overset{\tau_s}{\:\text\squiggleLR\!}} s_{i^*}
{\overset{\sigma}{\:\text\squigglea\!}} z_{i^*} {\overset{\tau_z}{\leftsquigarrow}} z_0 {\overset{\pi_z}{\leftsquigarrow}} z$
has length $d_{\calT}(S,Z)+2i^*+|\sigma|$.   

To bound $|\pi|$, consider the path
$\Pi$ from $S$ to $Z$ in $\calT$, which goes through $A_0$\changed{,}
and let $A_1,\dots,A_{i^*}$ be the ancestors of $A_0$ that were visited by the detour-method.   
If $i^*>0$, then 
by Obs.~\ref{obs:pathWithinNode}
we have $|V(A_j)|\geq 2(\xi_j{+}1)=2g$ for $j=1,\dots,i^*{-}1$,
node $A_{i^*}$ and the interior nodes of $\Pi$ store at least $g$ vertices each which nodes $Z$ and $R$ store at least one vertex.   Therefore
$$\tfrac{n}{2} \geq |V(\calT_R)| \geq 
2g(i^*{-}1) 
+g
+ (|\Pi|{-}1)g
+1+1
\ge g(|\Pi|+2i^*-2)+2,$$
hence $|\Pi|+2i^* \leq \tfrac{n{-}4}{2g}+2$ which is at most $\delta{+}1$ by integrality.
If $i^*=0$ then $|\Pi|+2i^* \leq \delta$ by Obs.~\ref{obs:distT}.
Either way $|\pi|=|\Pi|+2i^*+|\sigma|\leq \delta+1+ (g{-}1)$.
\qed\end{proof}

\paragraph{Vertices at descendants of $S$:}
In light of Claim~\ref{cl:caseRoot}, we only need to 
worry about vertices that are stored at descendants of $S$.   
We first introduce two methods to find short paths for these in special situations.   Recall that $\myalpha(S)$ denotes the number of vertices stored at strict ancestors of \changed{node} $S$.

\begin{numclaim}
\label{cl:alphaSmall}
\label{cl:smallAlpha}
Assume that $\myalpha(S)\leq g^2$ and $g\geq 3$.
Then  for any $s\in V(S)$ and any vertex $z$ stored
at a descendant $Z$ of $S$ we have
$d_H(s,z)\leq \delta+2g-2$.
\end{numclaim} 
\begin{proof}
See also Figure~\ref{fig:method2}.
Use directed edges to find a path $\pi_z$ of length $d_\calT(Z,S)\leq \delta$ from $z$ to some vertex $z_0\in V(S)$.
We know $d_H(s,z_0)\leq \max\{2\lfloor \sqrt{\myalpha(S)}\rfloor {-}2,4\} \leq 2g-2$ by Lemma~\ref{lem:connectS}.
Combining the paths gives the desired length: $d_H(z,s)\leq d_H(z,z_0)+d_H(z_0,s)\leq \delta+2g-2$.
\qed\end{proof}

\begin{numclaim}
\label{cl:smallSwitcher}
Let $D$ be a descendant of $S$.   
Let $s_D\in V(D)$ be a vertex of eccentricity at most  $\lfloor \tfrac{|V(D)|}{2} \rfloor$ in the
connected graph $H[V(D)]$ (Obs.~\ref{obs:radiusUpper}).   Follow directed edges from $s_D$  to reach a  \changed{vertex} $s\in V(S)$.
Then for any vertex $z$ stored
at a descendant $Z$ of $D$ we have $d_H(s,z)
\leq d_\calT(Z,S)+\lfloor \tfrac{|V(D)|}{2}\rfloor$.
\end{numclaim}
\begin{proof}
See Figure~\ref{fig:method3}.
%
Use directed edges to go from $z$ to a vertex $z_D\in V(D)$ along a path
$\pi_z$ of length $d_\calT(Z,D)$.   Find a path $\sigma$ of length at most $\lfloor \tfrac{|V(D)|}{2} \rfloor$ to connect $z_D$ to $s_D$,
and then follow $d_\calT(D,S)$ directed edges to get to $s$.
Combining the paths gives the desired length since $d_\calT(Z,D)+d_\calT(D,S)=d_\calT(Z,S)$.
\qed\end{proof}

\begin{corollary}
\label{cor:smallS}
\label{cl:smallS}
Assume that $|V(S)|\leq 4g-3$.
Then there exists an $s\in V(S)$ such that for any vertex $z$ stored
at a descendant of $S$ we have $d_H(s,z)\leq \delta+ 2g-2$.
\end{corollary}

So we are done if $\myalpha(S)\leq g^2$ or $|V(S)|\leq 4g-3$.
Otherwise we distinguish the descendants of $S$ by their depth,
using Claim~\ref{cl:smallSwitcher} (for a carefully chosen $D$) for the `deep' ones
and the method of Claim~\ref{cl:alphaSmall} for the others.
Define the \emph{threshold-value} $\theta=\lceil \tfrac{n-\myalpha(S)}{2g}\rceil -1$,
call a node $Z$ \emph{deep} if it is a descendant of $S$ with $d_\calT(S,Z)\geq \theta$,
and call a vertex \emph{deep} if it is stored at a deep node.
A straightforward math manipulation gives the following upper bound.

\begin{restatable}{observation}{Thetabound}
\label{obs:notDeep}
\label{obs:theta}
We have $\theta \leq \delta - 2\lceil \sqrt{\myalpha(S)} \rceil + 2g+1$.
\end{restatable}
\begin{proof}
Recall that $\delta=\lfloor \tfrac{n-2}{2g}\rfloor +1$ is an integer and observe that $\delta\geq \tfrac{n-1}{2g}$.
We also know that 
$$0\leq (\sqrt{\myalpha(S)}-2g)^2 = \myalpha(S)-4g\sqrt{\myalpha(S)}+4g^2
\quad  $$ $$\text{ or }\quad
\tfrac{\myalpha(S)}{2g} \geq 2\sqrt{\myalpha(S)}-2g> 2\lceil \sqrt{\myalpha(S)}\rceil -2g-2. $$
Therefore
\begin{align*}
\theta 
& = \lceil \tfrac{n-\myalpha(S)}{2g} \rceil  -1 \leq \tfrac{n-\myalpha(S)-1}{2g}  
 = \tfrac{n-1}{2g} - \tfrac{\myalpha(S)}{2g} < \delta - 2\lceil \sqrt{\myalpha(S)}\rceil +2g +2. 
\end{align*}
which yields the result by integrality. 
\qed\end{proof}
\begin{numclaim}
\label{claim:caseClose}
\label{cl:caseClose}
\label{claim:notDeep}
\label{cl:notDeep}
Assume that $g\geq 3$ and $|V(S)|\geq 4g-2$ and  $\myalpha(S)>g^2$.  
Then for any $s\in V(S)$ and any $z$
stored at a descendant $Z$ of $S$ that is not deep, 
we have $d_{H}(s,z)\leq \delta + 2g-2$.
\end{numclaim}
\begin{proof}
The method is exactly the same as in the proof of Claim~\ref{cl:alphaSmall} (walk from $z$ to a 
vertex $z_0\in V(S)$ and apply Lemma~\ref{lem:connectS} to connect $z_0$ to $s$, see also Figure~\ref{fig:method2}), 
but the analysis is different.   Since $Z$ is not deep, 
we have $d_H(z,z_0)=d_\calT(Z,S)\leq \theta{-}1 \leq \delta{-}2\lceil \sqrt{\myalpha(S)}\rceil {+}2g$. By $\myalpha(S)>g^2\geq 9$ 
we have $d_H(z_0,s)\leq 2\lceil \sqrt{\myalpha(S)}\rceil {-}2$ and so $d_H(z,s)\leq d_H(z,z_0)+d(z_0,s)\leq \delta+2g-2$.
\qed\end{proof}

For deep nodes we show that there exists a suitable node with which to apply Claim~\ref{cl:smallSwitcher}.
This is proved via a counting-argument: due to the (carefully chosen) threshold $\theta$ otherwise more than $n$ vertices would be stored in tree $\calT$. 
\tb{proof moved to the appendix}

\begin{numclaim}
\label{cl:smallCommonAncestor}
Assume that at least $4g-|V(S)|+1$ vertices are deep. 
Then there exists a descendant $D$ of $S$ (possibly $S$ itself) such that
$D$ is an ancestor of all deep nodes, and $|V(D)|\leq 2g-1$.  
\end{numclaim}
\begin{proof}  
Let $X$ be the least common ancestor of all deep nodes. This is a descendant of $S$ as well (possibly $X=S$), so
enumerate the path from $S$ to $X$ as $S{=}D_0,D_1,\dots,\allowbreak D_k{=}X$ for some $k\geq 0$.   We are done 
if $|V(D_i)|\leq 2g-1$ for some $0\leq i\leq k$, so assume (for contradiction) that the $k$ nodes $D_1,\dots,D_k$
store at least $2g$ vertices each.   If $X$ were deep (so $k\geq \theta$) then 
$D_1,\dots,D_{k}$ would store at least $\theta\cdot 2g \geq n-\myalpha(S)-2g$ vertices.   We also store $\myalpha(S)$ vertices at
strict ancestors at $S$, $|V(S)|$ vertices at $S$ and at least $4g-|V(S)|+1$ deep vertices, in total hence
more than $n$, impossible.

So $X$ is not deep, see also Figure~\ref{fig:leastCommonAncestor}.
Since $X$ is the least common ancestor of deep descendants, therefore it must have at least two children $X^1,X^2$
that are ancestors of deep descendants.   For $j=1,2$, enumerate the path from $X^j$ to a deep descendant
of $X^j$ as $D^j_{k+1},\dots,D^j_\theta$. Then for $i=k{+}1,\dots,\theta{-}1$
node $D_i^j$ is not a leaf and stores at least $g$ vertices, so $|V(D_i^1)|+|V(D_i^2)|\geq 2g$.   
So for $i=1,\dots,\theta-1$ we can find $2g$ vertices stored at not-deep descendants of distance $i$ from $S$.
In total these not-deep strict descendants of $S$ hence store at least
$(\theta{-}1)2g\geq n-\myalpha(S)-4g$ vertices.   As above therefore more than $n$ vertices are stored in the tree of peels, impossible.
\qed\end{proof}

Now we put everything together.

\begin{theorem}
\label{thm:main}
Let $G$ be a planar graph with $n\geq 3$ vertices and let $g\geq 3$ be an integer that is at most the fence-girth of $G$.
Then $G$ has a planar supergraph $H$ with $\rad{H}\leq \lfloor \tfrac{n-2}{2g} \rfloor +2g-1$.
Furthermore, a vertex of $H$ with this eccentricity can be found in linear time.
\end{theorem}
\begin{proof}
Compute layers $L_0,L_1,\dots$, augmentation $H$, and tree of peels $\calT$ as in Section~\ref{sec:toolbox}.
Find separator-node $S$ and compute $|V(S)|$, $\myalpha(S), \theta$,
and the number of deep vertices.
If $|V(S)|\geq 4g{-}2$ and there are at least $4g{-}|V(S)|{+}1$ deep vertices, then arbitrarily fix a deep node $Z$ and find node~$D$ of Claim~\ref{cl:smallCommonAncestor}
by walking from $S$ towards $Z$ until we encounter a node that stores at most $2g-1$ vertices.
In all other cases set $D:=S$.   Pick $s$ as in Claim~\ref{cl:smallSwitcher} applied to $D$.   Each of these steps takes linear time
(we elaborate on this in Section~\ref{app:linear}).

Applying various cases we show that $d_H(s,z)\leq\delta+2g-2$ for all $z$ (which implies the result by $\delta=\lfloor \tfrac{n-2}{2g} \rfloor +1$).
This holds for all $z\in V(\calT_R)$ by Claim~\ref{cl:caseRoot}, so consider a vertex $z$ stored at a descendant $Z$ of $S$.
The bound holds by Claim~\ref{cl:alphaSmall} if $\myalpha(S)\leq g^2$ and by Corollary~\ref{cl:smallS} if $|V(S)|\leq 4g-3$, so
assume neither.    If $Z$ is not deep, then apply Claim~\ref{cl:notDeep}.
If $Z$ is deep and there are at least $4g-|V(S)|+1$ deep vertices, then combine Claim~\ref{cl:smallCommonAncestor} with Claim~\ref{cl:smallSwitcher}
to get the bound.  
The only remaining case is that $\myalpha(S)>g^2$, $|V(S)|\geq 4g-2$, $Z$ is deep, 
and at most $4g-|V(S)|$ vertices are deep (so $|V(S)|\leq 4g-1$ since there is a deep vertex at $Z$).
In this case, $d_\calT(Z,S)=\theta$, for otherwise $Z$'s parent would also be deep and store $g\geq 3\geq 4g-|V(S)|+1$ deep vertices.
Also $\theta\leq \delta-2\lceil \sqrt{\myalpha(S)} \rceil +2g+1 \leq \delta-1$ by $\lceil \sqrt{\myalpha(S)} \rceil \geq g+1$.
We used $D=S$ and picked $s\in V(S)$ as in  Claim~\ref{cl:smallSwitcher}, so this 
gives $d_H(z,s)\leq d_T(Z,S)+\lceil \tfrac{V(S)}{2} \rceil \leq \delta-1+2g-1$ as desired.
\qed\end{proof}

Theorem~\ref{thm:main} implies {\bf C2} from the introduction, for if $G$ is triangulated then necessarily $H=G$ and so the
radius-bound holds for the input-graph $G$ as well.      It also implies {\bf C1}:

\begin{corollary}
\label{cor:main}
Let $G$ be a spherically-embedded graph with $n\geq 3$ vertices and let $g\geq 3$ be an integer that is at most the fence-girth of $G$.
Then $G$ has fse-outerplanarity at most $\lfloor \tfrac{n-2}{2g} \rfloor +2g$.
Furthermore, an outerface of $G$ that achieves this outerplanarity can be found in linear time.
\end{corollary}
\begin{proof}
Without changing the spherical embedding of $G$,
compute the super-graph $H$ with $\rad{H}\leq 
\lfloor \tfrac{n-2}{2g} \rfloor + 2g-1$;
the outerplanarity bound holds by Obs.~\ref{obs:simpleBound}
and the outerface can be found by picking any face incident to the vertex $s$ that achieves this eccentricity.
\qed\end{proof}

\paragraph{Discussion:}
Theorem~\ref{thm:main} requires $g\geq 3$.   We can always choose such a $g$ if $G$ is simple, but if $G$ has parallel edges
or loops (which we did not exclude) then in the fixed spherical embedding the \girth may only be 2 or 1.   We hence briefly
discuss the case $g\in \{1,2\}$.
Going through all proofs where $g\geq 3$ or $|V(N)|\geq 3$ is actually used
(Lemma~\ref{lem:distanceS}, Claim~\ref{cl:caseRoot} and~\ref{cl:smallAlpha}), one sees that the results hold
for $g=2$ if we increase the permitted distance-bound by 2.   (Likewise we need to raise the permitted
distance-bound in Claim~\ref{claim:notDeep} since it uses Lemma~\ref{lem:distanceS}.) 
Therefore $\rad{H} \leq \delta+2g=\lfloor \tfrac{n{-}2}{2g} \rfloor +1$ for $g=2$.   For $g=1$, $\rad{H}\leq 
\lfloor \tfrac{n}{2} \rfloor = \lfloor \tfrac{n-2}{2g}\rfloor +1$ by Obs.~\ref{obs:radiusUpper}.
\todo[color=green]{this paragraph could go if we need space}

Theorem~\ref{thm:main} also assumes that we are \emph{given} $g$.   However, the computation of $\calT$ does not depend on $g$,
and the only thing we require is that all its interior nodes store at least $g$ vertices.    So if we are not given $g$, then we can compute 
$\calT$, define $g^*=\min\{|V(N)|: N \text{ is an interior node}\}$ and use $g:=\min\{g^*,\sqrt{n-2}/2\}$ as parameter for Theorem~\ref{thm:main}.


\subsection{Run-time considerations}
\label{app:linear}

\tb{this entire subsection is new and not yet very polished}
In this section, we elaborate on why the various steps of our algorithm can be implemented in linear time.  No complicated data structures are needed for this; all bounds can be obtained via careful accounting of the visited edges.

Our first step is to compute the layers $L_0,L_1,\dots$, given the spherical embedding and the root-vertex $r$.   This can simply be done with a breadth-first search as follows.   Temporarily compute the \emph{radial graph}, which is a bipartite with one vertex class the vertices of $G$ and the other vertex class with one vertex per face of the planar embedding; it has edges whenever a face is incident to a vertex.   Then the layers $L_0,L_1,\dots$ are the same as the even-indexed BFS-layers of the radial graph if we start the breadth first search at root $r$.   Clearly this takes linear time to compute.

Next we must compute the augmentation $H$.   It follows directly from the proof in \cite{Bie15} that this can be done in linear time by scanning each face, but for completeness' sake we repeat this proof and analyze the run-time here.

\begin{numclaim} The augmentation $H$ can be computed in linear time.
\end{numclaim}
\begin{proof}
We first paraphrase the proof from \cite{Bie15} to show that graph $H$ exists.   Consider any face $F$, and fix one vertex $w$ of $F$ that minimizes its layer-number (i.e., the index $i$ of the layer $L_i$ containing $w$). Break ties among choices for $w$ arbitrarily.   For any vertex $v\neq w$ on $F$ that is not in $L_i$, add an edge $(v,w)$ if it did not exist already.    Clearly this maintains planarity since all new edges can be drawn inside face $F$. Repeat at all faces to get the edges for graph $H$.   To see that this satisfies the condition, consider an arbitrary vertex $v$ in layer $L_i$ for some $i>0$.   Thus $v$ was on the outer-face of $G\setminus (L_0\cup \dots \cup L_{i-1})$, but \emph{not} on the outer-face of $G\setminus (L_0\cup \dots \cup L_{i-2})$.   It follows that some face $F$ incident to $v$ had vertices in $L_{i-1}$.   So our procedure added an edge from $v$ to some vertex in face $F$ that is in layer $L_{i-1}$.

To find the edges efficiently, we assume that every vertex stores its layer-number. For each face $F$ we can then walk along $F$ in $O(\deg(F))$ time to find one vertex $w$ with the smallest layer-number.   In
a second walk along $F$,  add all edges $(v,w)$ of $H$ that fall within face $F$.   The only non-trivial step is to check
whether $(v,w)$ already existed.    This could be done with suitable data structures in $O(1)$ amortized time per edge, but the simplest approach is to not check this at all; duplicate edges in $H$ do not hurt us since we add at most $\deg(F)$ edges per face and hence a linear number of edges in total.
\qed\end{proof}

Our next step is to compute the tree  of peels $\calT$ of $H$, for which the main challenge is to compute the connected components after we have deleted some layers $L_0,\dots,L_i$. 
Assume we have kept track of all edges $E_i$ that connect $L_i$ to $L_{i+1}$.   For each $(y,z)\in E_i$ (say with $z\in L_{i+1}$),
walk along the face $F$ to the right of $(y,z)$, from $z$ and away from $y$, until we reach another edge $(y',z')\in E_i$ for which $F$ is to the left. Repeat at the face to the right of $(y',z')$, and continue repeating until we return to edge $(y,z)$ at the face to its left.   The visited vertices form the outer-face of one connected component of $G\setminus (L_0\cup \dots \cup L_i)$, so define a new node $N$ for them, set $V(N)$ to be the visited vertices, and make $N$ the child of the node that stored $y$.   
If there are edges of $E_i$ left that have not been visited yet, then repeat at them to obtain the next node.
At the end we have determined all nodes  that together cover layer $L_{i+1}$.    The run-time for this is proportional to $|L_{i+1}|+|E_i|$, so linear over all layers.

The next few steps (in the proof of Theorem~\ref{thm:main}) are to compute a number of values, and to traverse $\calT$ to determine all deep nodes and the number of vertices that they store; clearly this can be done in $O(|\calT|)$ time.   Likewise we can find the appropriate node $D$ to use in linear time, and finding $s$ can then be done in $O(|V(D)|+|\calT|)$ by using Observation~\ref{obs:radiusUpper} to find a central node $s_D$ in $D$ and then following outgoing edges until we reach $s\in S$.
The rest of the proof is an argument that the radius is small if we use $s$ as center, but we do not actually need to perform any computation here since we already
found the appropriate $s$.   So overall the run-time for Theorem~\ref{thm:main} is linear, and similarly one argues the run-time for Corollary~\ref{cor:main}.

\subsection{Tightness} 
We now design graphs with large outerplanarity (relative to the \girth $g$).  (These graphs
actually have \emph{girth} $g$, i.e., \emph{any} cycle (not just those that are fences) has 
length at least $g$.) We do this
first for the fixed-spherical-embedding outerplanarity, where the lower bounds hold
even for $g\in \{1,2\}$.   Then, for $g\geq 3$ and at a slight decrease of the lower bound, 
we give bounds for the (unrestricted) outerplanarity.
Roughly speaking, the graphs consist of nested cycles of length $g$, with a single
vertex or path inside the innermost / outside the outermost cycle, and (if desired) with edges added to ensure
that all spherical embeddings have the same fse-outerplanarity.   See Figure~\ref{fig:lb}.

\tb{FYI: defn of nested cycles kicked to appendix}

\begin{restatable}{lemma}{LowerFixed}
\label{lem:lowerFixed}
For $g\geq 1$ 
there exists an infinite class $\mathcal{G}_g=\{G_g^k: k\geq 1 \text{ odd}\}$ of spherically embedded graphs of girth and \girth $g$ for
which the fse-outerplanarity is at least
$\tfrac{k+3}{2}=\tfrac{n-2}{2g}+\tfrac{3}{2}$.
\end{restatable}
Before giving this proof, we briefly recall the definition of nested cycles.
Let $\mathcal{C}=\langle C_0,C_1,\dots,C_k,C_{k+1}\rangle$  be a sequence of disjoint subgraphs in a plane graph $G$.
We call $\mathcal{C}$ \emph{nested cycles} if for all $1\leq i\leq k$ subgraph $C_i$ is a cycle that contains
$C_0,\dots,C_{i-1}$ inside and $C_{i+1},\dots,C_{k+1}$ outside.    Note that $C_0$ and $C_{k+1}$ need not be cycles.

\begin{proof}
The graphs in $\mathcal{G}_g$ consist of nested $g$-cycles, with singletons as the
innermost and outermost `cycles'.  Formally, define $C_i$ (for $i=1,\dots,k$)
to be a $g$-cycle, set $C_0$ and $C_{k+1}$ to be singleton vertices,
and arrange $C_0,\dots,C_{k+1}$ as nested cycles to obtain the (disconnected) graph
$G_k^g$ with $n_k=gk+2$ vertices (see Figure~\ref{fig:G43}).

We claim that $G_k^g$ (for $k$ odd) has at least $\tfrac{k+3}{2}$ peels regardless of the
choice $F$ of the outerface.   To see this, let $i\in \{0,\dots,k\}$ be the index such that $F$ is
incident to $C_i$ and $C_{i+1}$; up to symmetry we may assume $i\geq k-i$ and therefore
$i\geq \lceil \tfrac{k}{2} \rceil =\tfrac{k+1}{2}$.
Let $L_1,L_2,\dots$ be the peels
when $F$ is the outerface.    Then $L_1$ contains $C_i$, but no vertex of $C_0\cup \dots \cup C_{i-1}$,
so we must have at least $i+1$ peels (containing $C_i,C_{i-1},\dots,C_1,C_0$).
Since $k=\tfrac{n_k{-}2}{g}$ therefore the number of peels is at least $i+1\geq \tfrac{k+3}{2} = \tfrac{n_k-2}{2g}+\tfrac{3}{2}$. 
\qed\end{proof}

\begin{restatable}{theorem}{Lower}
For $g\geq 3$ there exists an infinite class $\mathcal{H}_g=\{H_g^k: k\geq 3 \text{ odd}\}$ of planar graphs 
of girth and \girth $g$ that have outerplanarity at least $\tfrac{n-2}{2g} +1+\tfrac{3+\chi(\text{$g$ even})}{2g}$.
\end{restatable}
\begin{proof}
For $g=3$ the proof is very easy:
Take graph $G_k^3$ from Lemma~\ref{lem:lowerFixed} and arbitrarily triangulate it
while respecting the given spherical embedding. The resulting graph $H_k^3$
has a unique spherical embedding and requires (for $k$ odd) at least $\tfrac{n-2}{2g}+\tfrac{3}{2}=\tfrac{n-2}{2g}+1+\tfrac{3}{2g}$
peels. 
\medskip

For $g=4$ graph $H_k^4$ also extends $G_k^4$, but we must be more careful in how to add edges to
keep the \girth big.  Recall that $G_k^4$ consists of singleton $C_0$,
4-cycles $C_1,\dots,C_k$,  and singleton $C_{k+1}$, arranged as nested cycles.
Enumerate each $C_i$ as $\langle u_i,v_i,w_i,x_i \rangle$, where for $i=0,k{+}1$ all
four names refer to the same vertex.   Let $H_k^4$ be the graph obtained from $G_k^4$
by adding \emph{connector-edges} $(u_i,v_{i+1})$  and $(w_i,x_{i+1})$ for $i=0,\dots,k$ (see Figure~\ref{fig:H43}). 
One easily verifies that $H_k^4$ is bipartite, hence has \girth $4$.
It also is 2-connected and hence any spherical embedding can be 
achieved by permuting or flipping the 3-connected components at a cutting pair
\cite{DiBattista89a}.   But any 
cutting pair of $H_k^4$ has only two cut-components (hence we cannot permute), and
flipping the components gives the same graph (up to renaming) since $H_k^4$ is symmetric.
Therefore all spherical embeddings of $H_k^4$ are the same, up to renaming of vertices.
Hence for odd $k$ graph $H_k^4$ has at least $\tfrac{k+3}{2}=\tfrac{n-2}{2g} + \tfrac{3}{2} = \tfrac{n-2}{2g}+1+\tfrac{3+1}{2g}$
peels in any planar embedding. 
\medskip

\begin{figure}[ht]
    \hspace*{\fill}
\subcaptionbox{$G^4_3$\label{fig:G43}}{\includegraphics[page=1,width=.3\textwidth]{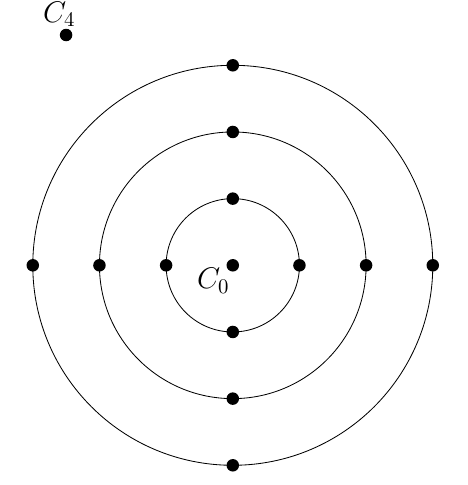}}
    \hspace*{\fill}
\subcaptionbox{$H_3^4$\label{fig:H43}}{\includegraphics[page=2,width=.3\textwidth]{pictures/lowerb.pdf}}
    \hspace*{\fill}
\subcaptionbox{$H^8_3$\label{fig:H83}}{\includegraphics[page=3,width=.3\textwidth]{pictures/lowerb.pdf}}
    \hspace*{\fill}
    \caption{ Graphs that have large outerplanarity relative to the \girth.  }
\end{figure}

Now consider $g\geq 6$ even, and assume that $H_k^{g-2}$ has been defined already.
To obtain $H_k^{g}$ from it, first extend path $C_0$ by one vertex, i.e., it becomes a path
with $\tfrac{g-2}{2}$ vertices from $u_0{=}v_0$ to $w_0{=}x_0$.   Likewise expand path $C_{k+1}$
by one vertex.   Finally for $1\leq i\leq k$, subdivide cycle $C_i$ twice, once on
the part between $u_i$ and $v_i$ and once on the part between $w_i$ and $x_i$.   
The connector-edges remain unchanged.   
Construct $H_k^{g-1}$ similarly from $H_k^{g-2}$: expand paths $C_0$ and $C_{k+1}$
by one vertex, but subdivide $C_i$ (for $1\leq i\leq k$) only once, on the part between $u_i$ and $v_i$.
For both $g$ even and odd, graph $H_k^g$ can be obtained by subdividing edges of $H_k^4$, hence
its outerplanarity cannot be better and is (for $k$ odd) at least
$\tfrac{k+3}{2}$.
Since $H_k^g$ has
$n=kg+g-2+\chi(\text{$g$ is odd})$ vertices, hence its outerplanarity is as desired.   It remains
to argue the girth, so fix an arbitrary simple cycle $C$ in $H_k^g$ and assume that it visits $C_i,\dots,C_j$
for some $i\leq j$ and no other nested cycles.    If $i=j$ then $C$ equals $C_i$ and has length $g$.   If
$i<j$, then $C$ uses at least two connector-edges, and parts of $C_i$ and $C_j$ that connect such connector-edges;
each such part has length at least $(g{-}1)/2$ and so $|C|>g$ in this case.   So the shortest cycle has
length at least $g$, and this is achieved (and the cycle is a fence) at $C_1$.
\qed\end{proof}

\newcounter{lb} 
\setcounter{lb}{\thefigure}

\begin{figure}[pt]
    \hspace*{\fill}
    \includegraphics[page=1,width=.3\textwidth]{pictures/lowerb.pdf}
    \hspace*{\fill}
    \includegraphics[page=2,width=.3\textwidth]{pictures/lowerb.pdf}
    \hspace*{\fill}
    \includegraphics[page=3,width=.3\textwidth]{pictures/lowerb.pdf}
    \hspace*{\fill}
    \caption{ Graphs $G^4_3$, $H^4_3$ and $H^8_3$.
	}
    \label{fig:lb}
\end{figure}

\section{Outerplanarity and Diameter} 
\label{sec:diameter}

With much the same techniques as for Theorem~\ref{thm:main}, we can also bound the outerplanarity in terms of the diameter, as long as we permit a `correction-term' of $O(\sqrt{n})$.

\begin{theorem}
\label{thm:diameter}
Any simple plane graph $G$ with $n\geq 14$ vertices has a plane supergraph $H$ with $\rad{H}\leq \lceil \tfrac{1}{2}\diam{G} \rceil + 2\sqrt{n{-}4}-2$.
\end{theorem}
\begin{proof}
Compute layers $L_0,L_1,\dots$, augmentation $H$ and the tree of peels $\calT$ as in Section~\ref{sec:toolbox}.
If $\diam{\calT}=1$ then $\calT$ consists only of root $R$ and its unique child $N$ that stores all vertices
of $V(G)\setminus \rho$; in consequence $G$ has outerplanarity at most $2\leq 1+2\sqrt{n{-}4}-1$ by $n\geq 5$.
So assume that $\diam{\calT}\geq 2$ and let $U,V$ be two nodes of $\calT$ with $d_\calT(U,V)=\diam{\calT}$.
Let $S$ be the node `halfway between them', i.e., of distance $\lceil \diam{\calT} /2 \rceil$ from $U$ along the
unique path from $U$ to $V$ in $\calT$.   One easily verifies that $d_\calT(S,Z)\leq \lceil \diam{\calT} /2 \rceil$
for all nodes $Z$ of $\calT$, otherwise $Z$ would be too far away from either $U$ or $V$ since $\calT$ is a tree.    Also note that
$\diam{\calT}\leq \diam{H}$ by Obs.~\ref{obs:Tproperty}(\ref{it:Tproperty}) and $\diam{H}\leq \diam{G}$
since $H$ is a supergraph.    Finally observe that $\myalpha(S)\leq n{-}4$, for $S$ is an interior node by $\diam{\calT}\geq 2$
and stores at least three vertices by simplicity while at least one of $U,V$ is not an ancestor of $S$ and stores at least one vertex.

Pick $s\in S$ arbitrarily.      For any $z\in V(G)$ (say $z$ is stored at node $Z$), we find a path from $s$ to $z$ as follows
(see also Figure~\ref{fig:method2} and~\ref{fig:method3}):
Let $A_0$ be the least common ancestor of $S$ and $Z$ (quite possibly $A_0=S$ or $A_0=Z$).   Follow directed edges from $s$ and $z$
to reach vertices $s_0,z_0\in A_0$; the total number of these edges is $d_\calT(S,Z)\leq \lceil \tfrac{1}{2}\diam{\calT} \rceil \leq \lceil \tfrac{1}{2}\diam{G} \rceil$.   
Observe that $4\leq \lceil \sqrt{n{-}4}\rceil$ and also $\sqrt{\myalpha(A_0)}\leq \sqrt{\myalpha(S)}\leq \sqrt{n{-}4}$.
Using Lemma~\ref{lem:connectS} therefore $d_H(s_0,z_0)\leq 2\lceil \sqrt{n{-}4}\rceil -2$ and so $d_H(s,z)\leq  	\lceil \tfrac{1}{2} \diam{G} \rceil 
+2\sqrt{n{-}4}-2$.
\qed\end{proof}

Theorem~\ref{thm:diameter} implies {\bf C3} from the introduction: 
The fse-outerplanarity of $G$ is at most the fse-outerplanarity of $H$, which is at most $\rad{H}+1$. 
If $G$ is triangulated then necessarily $H=G$ and so $\rad{G}\leq \lfloor \tfrac{1}{2}\diam{G} \rfloor + O(\sqrt{n})$. 
Theorem~\ref{thm:diameterLower} implies that the `correction-term' of $O(\sqrt{n})$ cannot be avoided for the graph in Figure~\ref{fig:mesh1}. 

\begin{figure}
    \centering 
\includegraphics[scale=0.4,angle=90,page=2]{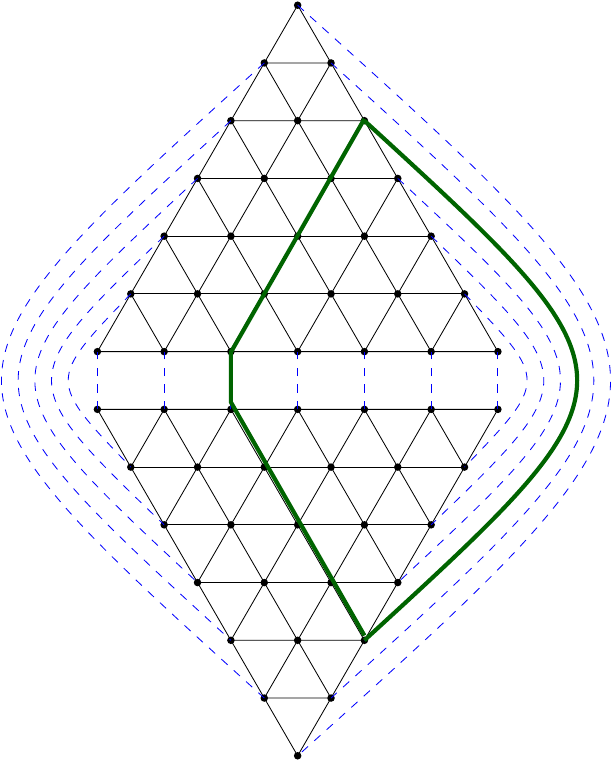} 
\hspace*{\fill} 
 \includegraphics[scale=0.4,page=3]{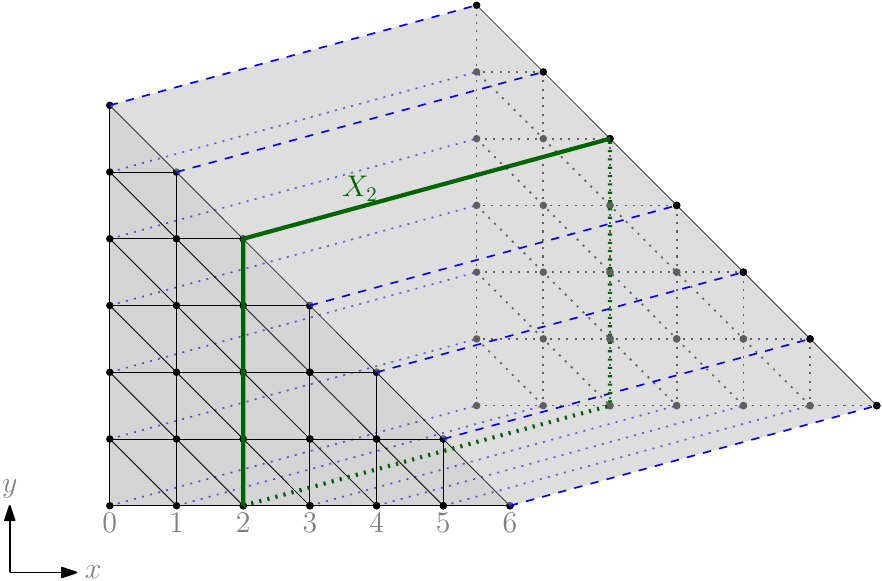} 
 \caption{Graph $H$ (for $k=2$) for the proof of Theorem~\ref{thm:diameterLower},
 both in a planar drawing and embedded  on the triangular prism.   
 Connector-edges are blue/dashed, $X_2$ is green/bold.  }
 \label{fig:mesh1}
 \end{figure}
 
\begin{restatable}{theorem}{Diameter}
For every positive integer $k$, there exists a triangulated graph $G$ with $n=(3k{+}1)(3k{+}2)$ vertices that has diameter at most $3k+1$ and radius at least  $2k=\tfrac{1}{2}\diam{G}+\Omega(\sqrt{n})$. 
\label{thm:diameterLower}
\end{restatable}

\begin{proof}
Let $M$ be the \emph{triangular grid of sidelength $3k$} defined as follows.   Each vertex of $M$ corresponds to a point $(x,y,z)$ in $\mathbb{Z}^3$ that satisfies $x,y,z\geq 0$ and $x+y+z=3k$.   Two such points are connected if and only if their Euclidean distance is $\sqrt{2}$, i.e., one of the three coordinates has changed by ${+}1$ while another has changed by ${-}1$.  Let $M'$ be a second copy of this grid, and add \emph{connector-edges} $(v,v')$ for any $v\in M$ and $v'\in M'$ that have the same coordinates, and one of these coordinates is 0.   We can visualize the resulting graph $H$ as lying on the triangular prism, after omitting the $z$-coordinates, see also Figure~\ref{fig:mesh1}.   Graph $H$ has $2\cdot \sum_{i=0}^{3k} (i{+}1) =(3k{+}1)(3k{+}2)$ vertices as desired.   It is not quite triangulated; let $G$ be obtained from $H$ by inserting arbitrary diagonals into the quadrangular faces incident to the connector-edges.

Define $X_{3k}$ to be the two vertices with $x$-coordinate $3k$.
For  $i=0,\dots,3k-1$, consider the set of all vertices that have $x$-coordinate $i$, and note that these form a cycle $X_i$ of length $6k-2i+2$. 
Using these cycles, it is very easy to lower-bound the radius.   Consider an arbitrary vertex $v$, say it has coordinates $(x,y,z)$.   Since $x+y+z=3k$ we may (up to renaming of coordinates) assume that $x\leq k$.  Let $w\in X_{3k}$. Then each of the disjoint cycles $X_{k+1},X_{k+2},\dots,X_{3k-1}$ contains $v\in X_k$ on one side and $w\in X_{3k}$ on the other.   So any path from $v$ to $w$ must contain at least one vertex from each of these cycles, and $d_H(v,w)\geq 2k$.   So any vertex has eccentricity at least $2k$, and $\rad{H}\geq 2k$.

Now we upper-bound the diameter.   Fix two arbitrary vertices $v,v'$, and assume that they have $x$-coordinates $i$ and $j$ respectively; up to renaming $i\leq j$.   We can walk from $v\in X_i$ to some vertex $w\in X_j$ in $j-i$ steps, since for all $\ell<3k$ every vertex in $X_\ell$ has at least one neighbour in $X_{\ell+1}$. So $d_H(v,w)\leq j-i$.   Vertices $w$ and $v'$ both belong to $X_j$, a cycle of length $6k-2j+2$, and hence $d_H(w,v')\leq 3k-j+1$.   Therefore $d_H(v,v')\leq (j{-}i)+(3k{-}j{+}1)\leq 3k+1$ and since this holds for all vertex-pairs we have $\diam{H}\leq 3k+1$.
\qed\end{proof}

\section{Remarks}

While the `$\tfrac{n}{2g}$'-part of our bound in Theorem~\ref{thm:main} is tight, the `$+2g$' part could use improvement.   We can easily prove (with the same techniques as in Theorem~\ref{thm:diameter}) a bound of $\tfrac{n-2}{2g}+O(\sqrt{n})$, but does every planar graph
with fence-girth $g$ have outerplanarity $\tfrac{n}{2g}+O(1)$?   

Also, our linear-time algorithm carefully side-steps the question of how to compute the fence-girth (it instead uses 
\changed{a} parameter $g$
for which the node-sizes of $\calT$ are big enough).  Testing whether the fence-girth is at most $k$ is easily done if the spherical 
embedding is fixed and $k$ is a constant, using the subgraph isomorphism algorithm by Eppstein \cite{DBLP:journals/jgaa/Eppstein99}.
But the fence-girth need
not be constant 
and Eppstein's algorithm does not work if the embedding can be changed.  Algorithms to compute the girth \cite{DBLP:journals/siamcomp/ChangL13}
do not seem transferrable  to the fence-girth.  How easy is it to compute the fence-girth, both when the spherical embedding is fixed and
when it can be chosen freely?

\section*{Acknowledgments}
Research by TB supported by NSERC; FRN RGPIN-2020-03958. Research by DM  supported by NSERC; FRN RGPIN-2018-05023.
 
\bibliographystyle{splncs04}
\bibliography{paper}

\end{document}